\documentclass[11pt]{article} 
\usepackage[utf8]{inputenc} 
\usepackage[affil-it]{authblk}
\usepackage{tikz} 
\usepackage{amsmath} 
\usepackage{dsfont} 
\usepackage{amsthm}
\usepackage{hyperref}
\usepackage{fullpage}
\usepackage{paralist}
\usepackage[backgroundcolor=blue!30!white]{todonotes}
\makeatletter
\presetkeys%
    {todonotes}%
    {inline}{}%
\makeatother

\newtheorem{theorem}{Theorem} 
\newtheorem{lemma}{Lemma} 
\newtheorem{definition}{Definition}
\newtheorem{corollary}{Corollary}
\newtheorem{proposition}[theorem]{Proposition}
\theoremstyle{remark}
\newtheorem{construction}[theorem]{Construction}

\newcommand{\pivd}{\textsc{$\Pi$-Vertex Deletion}} 
\newcommand{\ppivd}{\textsc{Planar $\Pi$-Vertex Deletion}} 
\newcommand{\cpivd}{\textsc{Connected $\Pi$-Vertex Deletion}} 
\newcommand{\cppivd}{\textsc{Connected Planar $\Pi$-Vertex Deletion}} 
\newcommand{\vc}{\textsc{Vertex Cover}} 
\newcommand{\pvc}{\textsc{Planar Vertex Cover}} 
\newcommand{\scvc}{\textsc{Subcubic Girth-$d$ Vertex Cover}}
\newcommand{\sctvc}{\textsc{Subcubic Girth-$3d$ Vertex Cover}}

\newcommand{\decprob}[3]{%
    \begin{quote}
        \textsc{{#1}}
    \begin{compactdesc}
    \item[Input:] #2
    \item[Question:] #3
    \end{compactdesc}
  \end{quote}

  }

\usepackage{etoolbox}

\newcommand{\appendixproof}[2]{
#2
}

\begin{document}

\title{Tight Running Time Lower Bounds for Vertex~Deletion Problems\footnote{Supported by the DFG, project KO~3669/4-1.}}
\author{Christian Komusiewicz}
\affil{Institut f\"ur Informatik, Friedrich-Schiller-Universit\"at Jena\\Ernst-Abbe-Platz 2, D-07743 Jena\\\texttt{christian.komusiewicz@uni-jena.de}}

\maketitle  

\begin{abstract} For a graph class~$\Pi$, the~\pivd{} problem has as
  input an undirected graph~$G=(V,E)$ and an integer~$k$ and asks
  whether there is a set of at most~$k$ vertices that can be deleted
  from~$G$ such that the resulting graph is a member of~$\Pi$.  By a
  classic result of Lewis and Yannakakis~[J.~Comput.~Syst.~Sci.~'80],
  \textsc{$\Pi$-Vertex Deletion} is NP-hard for all hereditary
  properties~$\Pi$. We adapt the original NP-hardness construction to
  show that under the Exponential Time Hypothesis (ETH) tight
  complexity results can be obtained. We show that \pivd{} does not
  admit a~$2^{o(n)}$-time algorithm where~$n$ is the number of
  vertices in~$G$. We also obtain a dichotomy for running time bounds
  that include the number~$m$ of edges in the input graph: On the one
  hand, if~$\Pi$ contains all independent sets, then there is
  no~$2^{o(n+m)}$-time algorithm for~\pivd{}. On the other hand, if
  there is a fixed independent set that is not contained in~$\Pi$ and
  containment in~$\Pi$ can determined in~$2^{O(n)}$ time or~$2^{o(m)}$
  time, then~\pivd{} can be solved in~$2^{O(\sqrt{m})}+O(n)$
  or~$2^{o({m})}+O(n)$~time, respectively. We also consider
  restrictions on the domain of the input graph~$G$. For example, we
  obtain that \pivd{} cannot be solved in~$2^{o(\sqrt{n})}$ time
  if~$G$ is planar and~$\Pi$ is hereditary and contains and excludes
  infinitely many planar graphs. Finally, we provide similar results
  for the problem variant where the deleted vertex set has to induce a
  connected graph.
\end{abstract}

\section{Introduction}
\label{sec:intro}

In graph modification problems, the aim is to modify a given
graph such that it fulfills a certain property~$\Pi$, for example
being acyclic, bipartite, having bounded diameter, or being an independent set (that is, an edgeless graph). Formally, a
graph property is simply a set of graphs. The focus is usually on
nontrivial graph properties, that is, properties such that the set of
graphs fulfilling the property and the set of graphs \emph{not}
fulfilling the property are infinite. Throughout this work, we
consider only nontrivial graph properties.

The most common operations to modify the graph are vertex deletion,
edge deletion and edge addition. Graph modification problems have many
applications, for example in data anonymization~\cite{Nic14}, data
clustering~\cite{BB13}, data visualization~\cite{BBB+12}, in the
modeling of complex networks via differential
equations~\cite{FMK+13}, and in social network analysis~\cite{BBH11}.

In this work, we consider graph modification problems where only
vertex deletions are allowed. Formally, we study the following
problem: \decprob{\pivd}{An undirected graph~$G=(V,E)$ and an
  integer~$k$.}{Is there a set~$S\subseteq V$ such that~$|S|\le k$
  and~$G[V\setminus S]$ is contained in~$\Pi$?}  

A natural class of graph properties occurring in many applications are
the hereditary graph properties: A graph property is \emph{hereditary}
if it is closed under vertex deletion. Equivalently, if a graph~$G$
fulfills a hereditary property~$\Pi$, then every \emph{induced
  subgraph} of~$G$ fulfills~$\Pi$. The family of hereditary graph
properties includes all monotone graph properties (which are the
properties closed under vertex and edge deletions) and all
minor-closed graph properties (which are the properties closed under
vertex and edge deletions and edge contractions).

Hereditary graph properties can be characterized by
forbidden induced subgraphs. That is, there is a (possibly infinite)
set of graphs~$\mathcal{F}$ such that~$G$ has property~$\Pi$ if and
only if it does not contain a graph from~${\cal F}$ as an induced
subgraph. For example, a graph is a forest if and only if it does not
contain any cycle as induced subgraph. For these graph properties,
there is the following classic hardness result due to
Lewis~and~Yannakakis~\cite{LY80}.
\begin{theorem}[\cite{LY80}]
  Let~$\Pi$ be a hereditary nontrivial graph
  property. Then~\textsc{$\Pi$-Vertex Deletion} is NP-hard.
\end{theorem}
While this result rules out the possibility of polynomial-time
algorithms, it leaves open the existence of subexponential-time
algorithms for these problems. Our aim is to rule out the existence of
subexponential-time algorithms for these problems by assuming the
Exponential Time Hypothesis (ETH)~\cite{IPZ01}. To this end, let us inspect the reduction by Yannakakis\footnote{The article by Lewis and Yannakakis explicitly attributes the construction which will form the basis of our reduction to Yannakakis~\cite{LY80}.} more closely.

The source problem in this reduction is \textsc{Vertex Cover}.
\decprob{\vc}{An undirected graph~$G=(V,E)$ and an integer~$k$.}{Is
  there a set~$S\subseteq V$ such that~$|S|\le k$ and~$G[V\setminus
  S]$ is an independent set?}  Assuming ETH, \vc{} does not admit
a~$2^{o(n+m)}$-time algorithm where~$n$ is the number of vertices in
the input graph and~$m$ is the number of edges. Hence, \vc{} is a
suitable source problem for a reduction in the ETH-framework. The
construction of Yannakakis, however, reduces $n$-vertex instances
of~\vc{} to~$\Theta(n^2)$-vertex instances of ~\pivd{}. This blowup in
the vertex number only yields a running time lower bound
of~$2^{o(\sqrt{n})}$. Thus, to obtain a better lower bound, we need to
``re-engineer'' the original construction.

We adapt the hardness construction of Yannakakis in two
ways: We take a slightly more fine-grained approach for choosing the
graphs which are used as basis for the gadgets in the reduction from the
\vc{} instance and we reduce from a more restricted
variant of \vc. These two modifications of the
reduction allow us to obtain our main result. Here,~$n$ denotes the
number of vertices in the input graph and~$m$ denotes its number of
edges.
\begin{theorem}\label{thm:main}
  Let~$\Pi$ be a hereditary nontrivial graph property, then:
  \begin{enumerate}
  \item If the ETH is true, then \pivd{} cannot be solved in~$2^{o(n+\sqrt{m})}$ time,
  \item If the ETH is true, then \pivd{} cannot be solved in~$2^{o(n+m)}$ time if~$\Pi$ contains
    all independent sets.
  \item If~$\Pi$ excludes some independent set, then \pivd{} can be
    solved in~$2^{O(\sqrt{m})}+O(n)$ time if and only if membership
    in~$\Pi$ can be recognized in~$2^{O(n)}$ time.
  \item If~$\Pi$ excludes some independent set, then \pivd{} can be
    solved in~$2^{o(m)}+O(n)$ time if and only if membership in~$\Pi$
    can be recognized in~$2^{o(m)}$ time.
  \end{enumerate}
\end{theorem}
The first and second statement of the theorem follow from our adaption
of Yannakakis' construction, described in
Section~\ref{sec:hardness}. The third and fourth statement follow from
a simple combinatorial algorithm, described in
Section~\ref{sec:subexp}. 

In addition to our main theorem, we observe that the reduction yields
a similarly tight lower bound for vertex deletion problems in planar
graphs. More precisely, we show that any nontrivial variant
of~\textsc{$\Pi$-Vertex Deletion} cannot be solved
in~$2^{o(\sqrt{n})}$ time on $n$-vertex planar graphs. This result
aligns nicely with the so-called square root phenomenon~\cite{Marx13}
that many NP-hard problems on planar graphs can be solved
in~$2^{O(\sqrt{n})}$ time but not faster. Moreover, we also obtain
tight running time lower bounds for input graphs with bounded degree
and bounded degeneracy (both depending on~$\Pi$) and for input graphs
containing a dominating vertex. Finally, we consider \cpivd{}, where
the vertex set that is deleted has to induce a connected graph. We
show that the bounds of Theorem~\ref{thm:main} also hold for this
problem variant.

\subparagraph{Related Work.}  Fellows et al.~\cite{FGMN11} adapted the
reduction of Yannakakis to provide NP-hardness results for the variant
of \pivd{} where one aims to find a solution that is disjoint from and
smaller than a given solution. The study considers the case that~$\Pi$
is closed under taking disjoint unions of graphs. These are exactly
the graph properties that have only connected minimal forbidden
induced subgraphs. As an introduction to their result, they describe a
modified version of the original reduction which already implies
Theorem~\ref{thm:main} for this case. In subsequent
work,~Guo~and~Shrestha~\cite{GS14} investigated the complexity of the
remaining case when the forbidden subgraphs have disconnected
forbidden subgraphs, showing NP-hardness for most cases.

When edge deletions or additions are the allowed operations, there is
no hardness result as general as the one of Lewis and
Yannakakis~\cite{LY80}. Nevertheless, many variants of
\textsc{$\Pi$-Edge Deletion} are NP-hard~\cite{Yan81}.  If~$\Pi$ has
exactly one forbidden induced subgraph, then a complexity dichotomy
and tight ETH-based running time lower bounds can be achieved for edge
deletion problems~\cite{ASS15a}. In similar spirit, tight running time
lower bounds have been obtained for some \textsc{$\Pi$-Completion} problems, where
one may only add edges to obtain the graph property~$\Pi$~\cite{DFPV15}.

\section{Preliminaries}
\label{sec:prelim}
Unless stated otherwise, we consider undirected simple graphs~$G=(V,E)$
where~$V$ is the vertex set and~$E\subseteq \{\{u,v\}\mid u,v \in V\wedge u\neq v\}$
is the edge set of the graph. We denote the vertex and edge set of a
graph~$G$ also by~$V(G)$ and~$E(G)$, respectively. The \emph{(open)
  neighborhood} of a vertex~$v$ is denoted by~$N(v):=\{u\mid
\{u,v\}\in E\}$ and the \emph{closed neighborhood} by~$N[v]:=N(v)\cup
\{v\}$. For a vertex set~$S\subseteq V$, let~$G[S]:=(S,\{\{u,v\}\in
E\mid u,v\in S\})$ denote the \emph{subgraph of~$G$ induced
  by~$S$}. For a vertex set~$S\subseteq V$, let~$G-S:=G[V\setminus S]$
denote the induced subgraph of~$G$ obtained by deleting the vertices
of~$S$ and their incident edges. Similarly, for a vertex~$v\in V$
let~$G-v:= G-\{v\}$ denote the graph obtained by deleting this vertex
and all its incident edges. A vertex~$v$ in a graph is a
\emph{cut-vertex} if~$G-v$ has more connected components than~$G$. A
graph is~\emph{$d$-degenerate} if every induced subgraph contains a
vertex of degree at most~$d$. Let~$A=(a_1,a_2, \ldots , a_n)$
and~$B=(b_1,b_2,\ldots , b_m)$ be two sequences where~$i$ is the
smallest index such that~$a_i\neq b_i$. Then,~$A$ is
\emph{lexicographically smaller than}~$B$ if~$a_i<b_i$.

We say that a parameterized problem~$(I,k)$ has a
\emph{subexponential-time algorithm} if it can be solved
in~$2^{o(k)}\cdot n^{O(1)}$ time. In this work, the parameters under
consideration are~$n$, the number of vertices in the input graph,
and~$m$, the number of edges in the input graph. The \emph{Exponential
  Time Hypothesis (ETH)} states essentially that the \textsc{3-SAT}
problem, which is given a boolean formula in 3-conjunctive normal form
and asks whether this formula can be satisfied by assigning truth
values to the variables, cannot be solved in~$2^{o(n)}$ time where~$n$
is the number of variables in the input formula~\cite{IPZ01}. For a
survey on lower bounds based on the Exponential Time Hypothesis, refer
to~\cite{LMS11}.

We now show that, assuming the ETH, the following special case of
\vc{} does not admit a subexponential-time algorithm. This problem
variant will be used as source problem in our reduction; here, a
subcubic graph is one with maximum degree at most three.
\decprob{\scvc}{An undirected subcubic $2$-degenerate graph~$G=(V,E)$ such that every
  cycle in~$G$ has length at least~$d$ and an integer~$k$.}{Does~$G$
  have a vertex cover of size at most~$k$?}  We show the running time
lower bound by devising a simple reduction from~\vc{} that replaces
edges by long paths on an even number of vertices. 
\begin{theorem}\label{thm:scvc}
  If the ETH is true, then \scvc{} cannot be solved in~$2^{o(n+m)}$
  time.
\end{theorem}
\appendixproof{Theorem~\ref{thm:scvc}}
{
\begin{proof}
  We reduce from \vc{} in subcubic graphs, which is known to admit
  no~$2^{o(n)}$-time algorithm (assuming ETH)~\cite{JS99}. Assume
  that~$d$ is even; otherwise, increase~$d$ by one. Given a vertex
  cover instance~$(G=(V,E),k)$ the reduction is simply to replace each
  edge~$\{u,v\}$ of~$G$ by a path on~$d$ vertices and to make one
  endpoint of this path adjacent to~$u$ and the other one adjacent
  to~$v$. Equivalently, subdivide each edge~$d$ times. Let~$G'$ denote
  the resulting graph, and call the paths that are added during the
  construction the~$d$-paths of~$G'$.  To complete the construction,
  set~$k':=k+(d/2)\cdot |E(G)|$. As we will show, the two instances
  are equivalent. Assume for now that equivalence holds. The shortest
  cycle in~$G'$ clearly contains more than~$d$ vertices (in fact it
  contains more than~$3d$ vertices). Moreover, the construction does
  not increase the maximum degree of~$G$, hence~$G'$ has maximum
  degree three, and each original edge was replaced by a path,
  hence~$G$ is 2-degenerate. Finally,~$G'$ has~$O(|V(G)|)$ many
  vertices and edges. Thus, any algorithm that solves~$(G',k')$
  in~$2^{o(|V(G')|+|E(G')|)}$ time can be used to obtain an algorithm
  that solves \vc{} on subcubic~$n$-vertex graphs in~$2^{o(n+m)}$ time
  which contradicts the ETH.

  Thus to complete the proof it remains to show equivalence of the
  instances.
  \begin{quote}
    $(G,k)$ is a yes-instance~$\Leftrightarrow$~$(G',k')$ is
    yes-instance.
  \end{quote}
  
  $\Rightarrow$: Let~$S$ be size-$k$ vertex cover of~$G$. Consider the
  graph~$G'-S$. By construction, all remaining edges are incident with
  vertices in~$d$-paths. Moreover, for each~$d$-path, the nonpath
  neighbor of at least one of its endpoints is contained in~$S$
  since~$S$ is a vertex cover. Thus, deleting the other endpoint of
  the path and then every other vertex of the path gives a set
  of~$d/2$ vertex deletions that together destroy all edges that are
  incident with vertices of this path. This can be done for
  all~$|E(G)|$ many~$d$-paths, resulting in~$(d/2)\cdot |E(G)|$ vertex
  deletions. Consequently,~$(d/2)\cdot |E(G)|+k$ vertex deletions
  suffice to destroy all edges in~$G'$.

  $\Leftarrow$: Observe that any vertex cover contains at
  least~$(d/2)\cdot |E(G)|$ vertices that are on~$d$-paths. Moreover,
  there is a minimum-cardinality vertex cover such that no two neighbors~$u$
  and~$v$ in a~$d$-path are deleted: Deleting~$u$ makes~$v$ a
  degree-one vertex which means that deleting the other neighbor
  of~$v$ instead of~$v$ gives a vertex cover which is at most as large
  as one that deletes~$u$ and~$v$. Thus, there is a minimum-cardinality
  vertex cover~$S'$ that contains for each~$d$-path exactly one of its
  two endpoints. Moreover, the neighbor of the other endpoint that is
  not in the~$d$-path is contained in~$S'$. Let~$S:=S'\cap V$ denote
  the vertices from~$V$ that are contained in~$S$. For each
  edge~$\{u,v\}\in E$, there is a~$d$-path connecting~$u$ and~$v$
  in~$E$. By the discussion above, either~$u$ or~$v$ is
  in~$S$. Thus,~$S$ is a vertex cover. Since~$S'$ contains at
  least~$(d/2)\cdot |E(G)|$ vertices from~$d$-paths, we have~$|S|\le
  k$.
\end{proof}
}
\section{Hardness for Graph Properties Containing all Independent
  Sets}
\label{sec:hardness}
The main idea of the reduction of Yannakakis is the following. Reduce
from \vc{} by replacing edges of the \vc{} instance by some graph~$H$
that is not contained arbitrarily often in any graph fulfilling~$\Pi$ (and
thus needs to be destroyed by a vertex deletion). Then, deleting
edges in the \vc{} instance corresponds to destroying forbidden
induced subgraphs in the \pivd{} instance. The main technical
difficulty arises when the forbidden subgraphs for graph
property~$\Pi$ have cut-vertices as in this case, two graphs that are
used to replace two edges incident with the same vertex may form
another forbidden subgraph of~$\Pi$. Thus, the graph~$H$ which is used
in the construction must be chosen carefully. To this end, Yannakakis
introduces the notion of~$\alpha$-sequence for connected graphs.
\begin{definition}
  Let~$H$ be a connected graph, let~$c$ be a vertex in~$H$, and
  let~$H_1, \dots , H_\ell$ denote the connected components of~$H-c$
  such that~$|V(H_1)|\ge |V(H_2)|\ge \dots
  |V(H_\ell)|$. Then,~$\alpha(H,c):=(|V(H_1)|,|V(H_2)|,\dots
  ,|V(H_\ell)|)$. The \emph{$\alpha$-sequence} of~$H$,
  denoted~$\alpha(H)$, is the lexicographically smallest sequence such
  that~$\alpha(H)=\alpha(H,c)$ for some~$c\in V(H)$.
\end{definition}
An example of a graph~$H$ with its~$\alpha$-sequence is presented in
Figure~\ref{fig:alpha}. 
\begin{figure}[t]
  \centering
  \begin{tikzpicture}
      \tikzstyle{edge} = [-,thick]

      \tikzstyle{vertex}=[circle,draw,fill=black,minimum
      size=6pt,inner sep=1pt,font=\footnotesize]

      \node[vertex] (1) at (0,0) {}; 
      \node[vertex] (3) at (0,2) {};
      \node[vertex] (4) at (1,0) {}; 
      \node[vertex,label=right:$\;c$] (5) at (1,1) {};
      \node[vertex,label=above:$d$] (6) at (1,2) {};
      \node[vertex] (7) at (2,0) {};
      \node[vertex] (9) at (2,2) {};

      \draw[edge] (1)--(4);
      \draw[edge] (1)--(5);       
      \draw[edge] (4)--(5);
      \draw[edge] (5)--(6);
      \draw[edge] (5)--(3);
      \draw[edge] (5)--(7);
      \draw[edge] (4)--(7);
      \draw[edge] (6)--(9);

      \path[fill=black,opacity=0.2,rounded corners] (-0.5,-0.2) -- (2.5,-0.2) -- (1,1.3) -- cycle;

    \end{tikzpicture}
    \caption{A graph~$H$ with $\alpha$-sequence~$(3,2,1)$. Deleting
      any non-cut-vertex~$v$ gives~$\alpha(H,v)=(6)$. Deleting the
      cut-vertex~$c$ or~$d$ gives~$\alpha(H,c)=(3,2,1)$
      and~$\alpha(H,d)=(5,1)$, respectively. The graph~$J(H)$ is highlighted by a
      gray background.}
  \label{fig:alpha}
\end{figure}
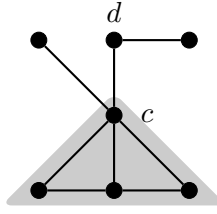
The idea in the reduction of Yannakakis is to choose a forbidden
subgraph that has a connected component with a lexicographically
smallest~$\alpha$-sequence as basis for the gadget.  More precisely,
the edges of the vertex cover instance are replaced by the largest
connected component remaining after deletion of a vertex~$c$ which
produces the smallest~$\alpha$-sequence of~$H$. To this end, let~$c\in
V(H)$ be a fixed vertex such that~$\alpha(H)=\alpha(H,c)$, and let~$J'$ be
a fixed largest connected component of~$H-c$. Then, let~$J(H):=H[V(J')\cup
\{c\}]$ denote the subgraph of~$H$ consisting of this component
plus~$c$. Since~$H$ has an edge,~$J(H)$ has at least two vertices. The
graph~$J(H)$ will be used to replace the edges of the \vc{} instance.

This terminology is sufficient to deal with the case that all
forbidden subgraphs are connected. To deal with disconnected forbidden
induced subgraphs, we introduce the notion
of~$\Gamma$-sequence.\footnote{Yannakakis uses the related notion
  of~$\beta$-sequence which does not break ties between nonisomorphic
  graphs that have the same~$\alpha$-sequence~\cite{LY80}. This
  tie-breaking, however, is necessary for our reduction. Guo and
  Shrestha~\cite{GS14} use~$\Omega$-sequences which also do not break
  ties between nonisomorphic graphs.} Observe that
the~$\alpha$-sequences imply a partial ordering on all graphs, but
there maybe nonisomorphic graphs that have the
same~$\alpha$-sequence. We want to avoid this, and thus the first step
is to refine this partial ordering to a total ordering by
breaking ties arbitrarily.
\begin{definition}
  A  function from the set of all connected graphs to~$\mathds{N}$ is a \emph{$\gamma$-ordering} if $\gamma(G_i)\ge \gamma(G_j)$ implies~$\alpha(G_i)\ge\alpha(G_j)$
  and~$\gamma(G_i)=\gamma(G_j)$ if and only if~$G_i$ and~$G_j$ are
  isomorphic.
\end{definition}
In the following, fix an arbitrary~$\gamma$-ordering of all graphs. To
obtain the~$\Gamma$-sequence we consider the sequence
of~$\gamma$-values created by the connected components of a graph~$H$.
\begin{definition}
  Let~$H$ be a graph with connected components~$H_1, \dots , H_\ell$
  such that~$\gamma(H_1)\ge \gamma(H_2) \ge \dots \ge
  \gamma(H_\ell)$. Then, the \emph{$\Gamma$-sequence} of~$H$
  is~$\Gamma(H):=(\gamma(H_1),\gamma(H_2), \dots ,\gamma(H_\ell))$.
\end{definition}
Now, the basis for our construction will be the minimal forbidden
induced subgraph of~$\Pi$ that has the lexicographically
smallest~$\Gamma$-sequence. For a graph property~$\Pi$, denote this
graph by~$H_\Pi$, an example is given in~Figure~\ref{fig:hpi}.
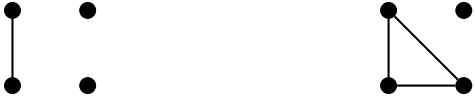
\begin{figure}[t]
  \centering
  \begin{tikzpicture}
      \tikzstyle{edge} = [-,thick]

      \tikzstyle{vertex}=[circle,draw,fill=black,minimum
      size=6pt,inner sep=1pt,font=\footnotesize]

      \node[vertex] (1) at (-1,0) {}; 
      \node[vertex] (2) at (-1,1) {};
      \node[vertex] (3) at (0,0) {};
      \node[vertex] (4) at (0,1) {};
      
      \draw[edge] (1)--(2);

      \node[vertex] (1c) at (4,0) {}; 
      \node[vertex] (2c) at (4,1) {};
      \node[vertex] (3c) at (5,0) {};
      \node[vertex] (4c) at (5,1) {};
      
      \draw[edge] (1c)--(2c);
      \draw[edge] (3c)--(2c);
      \draw[edge] (1c)--(3c);

    \end{tikzpicture}
    \caption{The two minimal forbidden subgraphs co-diamond (left) and
      co-claw (right) for~$\Pi$ being the
      complement graphs of line graphs of triangle-free graphs (see
      \url{http://www.graphclasses.org/}). Here,~$H_\Pi$ is the
      co-diamond since the lexicographically largest~$\alpha$-sequence of
      any component of the co-diamond is~$(1)$ and  the co-claw has a component with~$\alpha$-sequence~$(2)$.}
  \label{fig:hpi}
\end{figure}
We are now ready to prove the second statement of
Theorem~\ref{thm:main}. We reduce from~\sctvc{}. The construction of
the \pivd{} instance uses gadgets based on the graph~$H_\Pi$, an
example is presented in Figure~\ref{fig:construction}.
\begin{figure}[t]
  \centering
  \begin{tikzpicture}[yscale=0.8]
      \tikzstyle{edge} = [-,thick]

      \tikzstyle{vertex}=[circle,draw,fill=black,minimum
      size=6pt,inner sep=1pt,font=\footnotesize]

      \tikzstyle{vertexb}=[circle,draw,fill=gray,minimum
      size=6pt,inner sep=1pt,font=\footnotesize]

    \begin{scope}[shift={(-5,0)}]      
      \node[vertex] (1) at (0,0.5) {}; 
      \node[vertex] (2) at (1,0) {};
      \node[vertex] (3) at (1,1) {};
      \node[vertex] (4) at (2,1) {};
      \node[vertex] (5) at (2,0) {};
      \node[vertex] (6) at (3,0.5) {};
      \node[vertex] (7) at (4,0.5) {};
      \node[vertex] (4b) at (5,1) {};
      \node[vertex] (5b) at (5,0) {};
      \node[vertex,label=above:$c$] (6b) at (6,0.5) {};
      \node[vertex] (7b) at (7,0.5) {};
      
      \draw[edge] (2)--(3);
      \draw[edge] (4)--(5)--(6)--(4);      
      \draw[edge] (6)--(7);

      \draw[edge] (4b)--(5b)--(6b)--(4b);      
      \draw[edge] (6b)--(7b);
    \end{scope}
    \begin{scope}[shift={(5,0)}]      
      \node[vertexb] (1) at (0.3,0.5) {}; 
      \node[vertexb] (2) at (1,0) {};
      \node[vertexb] (3) at (1,1) {};
      \node[vertexb] (4) at (2,1) {};
      \node[vertexb] (5) at (2,0) {};

      \draw[edge] (1)--(2)--(3)--(4)--(5)--(2);
      \draw[edge] (1)--(3);
    \end{scope}
    \begin{scope}[shift={(0,-1.5)}]
      \foreach \x in {-2,...,7}
      \node[vertex] (1\x) at (-5+\x,-0.3) {}; 

      \foreach \x in {-2,...,7}{
      \node[vertex] (2\x) at (-5+\x,-1) {}; 
      \node[vertex] (3\x) at (-5+\x,-2) {}; 
      \draw[edge] (2\x)--(3\x);}

    \begin{scope}[shift={(2,-1.7)}]
      \node[vertex] (4) at (2,1) {};
      \node[vertex] (5) at (2,0) {};
      \node[vertex] (6) at (3,0.5) {};
      \node[vertex] (7) at (4,0.5) {};
      \draw[edge] (4)--(5)--(6)--(4);      
      \draw[edge] (6)--(7);
    \end{scope}

    \begin{scope}[shift={(-2,-6)}]
      \node[vertexb] (1) at (0.5,1) {}; 
      \node[vertexb] (2) at (2,0) {};
      \node[vertexb] (3) at (2,2) {}; 
      \node[vertexb] (4) at (4,2) {};
      \node[vertexb] (5) at (4,0) {};
      \node[vertex] (1b) at (-0.2,1) {}; 
      \node[vertex] (2b) at (2,0-0.7) {};
      \node[vertex] (3b) at (2,2+0.7) {}; 
      \node[vertex] (4b) at (4,2+0.7) {};
      \node[vertex] (5b) at (4,0-0.7) {};
      \foreach \x in {1,...,5}
      \draw[edge] (\x)--(\x b);

      \draw[edge] (1)--(2)--(3)--(4)--(5)--(2);
      \draw[edge] (1)--(3);

      \node[vertex] (12) at (0.9,0) {}; 
      \draw[edge] (1)-- (12)--(2); 

      \node[vertex] (13b) at (0.9,2) {}; 
      \draw[edge] (1)-- (13b)--(3); 

      \node[vertex] (23b) at (2.7,1) {}; 
      \draw[edge] (2)-- (23b)--(3); 
      
      \node[vertex] (34b) at (3,2.7) {}; 
      \draw[edge] (3)-- (34b)--(4); 
      \node[vertex] (45b) at (4.7,1) {}; 
      \draw[edge] (4)-- (45b)--(5); 
      \node[vertex] (52b) at (3,-0.7) {}; 
      \draw[edge] (5)-- (52b)--(2); 
     
    \end{scope}

    \end{scope}
  \end{tikzpicture}
  \caption{An example of the reduction (in this case from~\vc{}
    instead of~\scvc{}). Top left: the forbidden subgraph~$H_\Pi$ with
    the vertex~$c$ of~$H_1$; top right: the graph~$G$ of the \vc{}
    instance; bottom: the constructed graph~$G'$ of the~\pivd{}
    instance. Vertices of~$G$ and their corresponding vertices in~$G'$
    are shown in gray.}
  \label{fig:construction}
\end{figure}
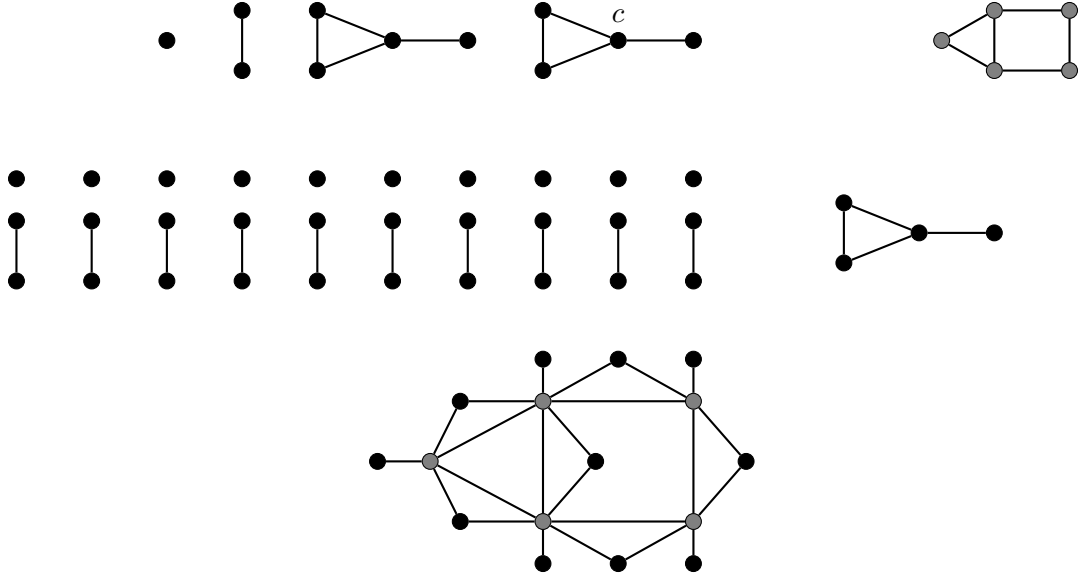

\begin{construction}
\label{cons:main}

Let~$(G=(V,E),k)$ be the given instance of~\sctvc{}. We build an instance
of~\pivd{} as follows. Let~$H_1$ be the connected component
of~$H_\Pi$ with maximum~$\gamma$-value and let~$d$ denote the number
of connected components of~$H_\Pi$ that are isomorphic
to~$H_1$. Moreover, let~$c$ denote a vertex of~$H_1$ such
that~$\alpha(H_1,c)=\alpha(H_1)$, that is,~$c$ is a vertex whose
deletion produces a lexicographically smallest sequence of connected
component sizes. Assume without loss of generality that~$|V|>|V(H_\Pi)|$ (otherwise, we can solve the \sctvc{}
instance in constant time).
  
  Starting from an empty graph~$G'$, construct the instance of~\pivd{}
  as follows. For each connected component~$H_i$ of~$H_\Pi$ that is
  not isomorphic to~$H_1$, add $2\cdot |V|$ disjoint copies of~$H_i$
  to~$G'$.\footnote{This is the part of the construction, where
    tie-breaking is necessary, as these graphs might have the
    same~$\alpha$-sequence as~$H_1$.} Then, add~$d-1$ disjoint copies
  of~$H_1$ to~$G'$. Call the graph constructed so far the \emph{base
    graph} of the construction.
  
  Now, add the graph~$G$ to~$G'$ and replace each edge~$\{u,v\}$
  of~$G$ by the graph~$J:=J(H_1)$ in the following way.\footnote{This
    part of the construction follows the construction of
    Yannakakis~\cite{LY80}.}  Let~$c'$ be a vertex from~$J$ that is
  different from~$c$. Remove the edge~$\{u,v\}$ from~$G'$, add a copy
  of~$J$ to~$G'$ and identify~$u$ with~$c$ and~$v$ with~$c'$. Now,
  let~$D:=H_1-(V(J)\setminus \{c\})$. For each vertex~$v\in V$, add a
  disjoint copy of~$D$ to~$G'$ and identify the vertex~$c$ of this
  copy with~$v$. Call the part of~$G'$ that is obtained from~$G$ the
  \emph{vc-extension} of the construction. Complete the construction
  of the \pivd{} instance by setting~$k':=k$.
\end{construction}
Using this construction, we can show our main running time lower
bound.
\begin{lemma}\label{lem:no-ind-set}
  Let~$\Pi$ be a hereditary graph property such that all independent
  sets are in~$\Pi$. Then, if the ETH is true,~\pivd{} cannot be
  solved in~$2^{o(n+m)}$ time.
\end{lemma}
\begin{proof}
  Let~$(G,k)$ be an instance of~\sctvc{} and let~$(G',k)$ be
  the~\pivd{} instance obtained from Construction~\ref{cons:main}. We
  first show the equivalence of the two instances, that is,
  \begin{quote}
    $G$ has a vertex cover of size at most~$k$
    $\Leftrightarrow$~$G'$ has a~$\Pi$-vertex deletion set of size at
    most~$k$.
  \end{quote}

  $\Rightarrow$: Let~$S$ be a size-$k$ vertex cover of~$G$. Then,~$S$
  is a~$\Pi$-vertex deletion set in~$G'$, that is,~$G'-S$
  fulfills~$\Pi$. To see this, first observe that each connected
  component~$C$ of the vc-extension of~$G'-S$ has~$\gamma$-value lower
  than~$\gamma(H_1)$: If~$C$ does not contain a vertex from~$V$, then
  it is either a proper induced subgraph of~$D$ or of~$J$. In both
  cases,~$C$ has a lexicographically lower~$\alpha$-sequence
  than~$H_1$ and thus lower~$\gamma$-value. If~$C$ contains a
  vertex~$v$ from~$V$, then~$v$ is a cut-vertex in~$C$, as deleting it
  disconnects the remainder of the copy of~$D$ containing~$v$. Now,
  observe that the other connected components of~$C-v$ have size at
  most~$|V(J)|-2$: These components are subgraphs of copies of~$J$
  corresponding to an edge~$\{u,v\}$ of $G$. Since~$S$ is a vertex
  cover,~$u$ is deleted from~$G$, cutting the rest of the graph from
  this copy of~$J$ and leaving only~$|V(J)-2|$ vertices in the
  component. Consequently, the~$\alpha$-sequence of~$C$ is lexicographically
  smaller than the~$\alpha$-sequence of~$H_1$. This implies that
  the~$\gamma$-value of the connected component~$C$ is smaller
  than~$\gamma(H_1)$.

  Moreover, the base part contains exactly~$d-1$ connected components
  with~$\gamma$-value~$\gamma(H_1)$. Hence,~$G'-S$ fulfills~$\Pi$: Any
  induced subgraph of~$G'-S$ has at most~$d-1$ connected components
  with~$\gamma$-value~$\gamma(H_1)$. By the assumption that~$H_\Pi$
  \begin{itemize}
  \item has exactly~$d$ connected components
    with~$\gamma$-value~$\gamma(H_1)$ and no connected components with larger~$\gamma$-value, and
  \item $H_\Pi$ has the lexicographically smallest~$\Gamma$-sequence among all
    forbidden induced subgraphs of~$\Pi$,
  \end{itemize}
  $G'-S$ cannot contain any forbidden induced subgraph of~$\Pi$ as
  induced subgraph.

  $\Leftarrow$: Let~$S'$ be a size-$k$~$\Pi$-vertex deletion set
  in~$G'$. Since~$k<2n$, we have that the number of copies of all
  connected components of~$H_\Pi$ with~$\gamma$-value smaller
  than~$\gamma(H_1)$ is at least as large in~$G'-S'$ as it is
  in~$H_\Pi$. Hence,~$G'-S'$ contains at most~$d-1$ vertex-disjoint copies
  of~$H_1$ as otherwise,~$G'-S'$ contains~$H_\Pi$ as induced
  subgraph. Then, create a
  set~$S_A\subseteq V$ as follows:
  \begin{itemize}
  \item For each vertex~$v\in V\cap S'$, add~$v$ to~$S_A$,
  \item for each copy~$D^*$ of~$D$ in~$G'$ such that~$S'$
    contains a vertex from~$V(D^*)\setminus
    V$, add the vertex~$v\in V(D^*)\cap
    V$ to~$S_A$,
  \item for each copy~$J^*$ of~$J$ such that~$S'\setminus V$ contains
    a vertex from~$V(J^*)\setminus V$ add an arbitrary vertex
    of~$V(J^*)\cap V$ to~$S_A$.
  \end{itemize}
  Observe that the set~$S_A$ has size at most~$k$. If~$G-S_A$ is an
  independent set, then~$G$ has a vertex cover of size at
  most~$k$. Otherwise, consider the graph~$G-S_A$ and observe that
  every edge~$\{u,v\}$ in~$G-S_A$ directly corresponds to an induced
  copy of~$H_1$ in~$G'-S'$: Since~$\{u,v\}$ is present in~$G-S_A$, we
  have that~$S'$ does not contain~$u$ or~$v$, does not contain
  vertices from the copies of~$D$ that are attached to~$u$ or~$v$, and
  does not contain a vertex from the copy of~$J$ attached to~$u$
  and~$v$. Now let~$q$ denote the size of a maximum matching in~$G-S_A$
  and observe that this implies that the vc-extension of~$G'-S'$
  has~$q$ vertex-disjoint copies of~$H_1$.
  Since~$G'-S'$ contains at most~$d-1$ vertex-disjoint copies
  of~$H_1$, this implies that~$q\le d-1$ and that~$S'$ contains at
  least~$q$ vertex deletions in the base graph (it needs to destroy at
  least~$q$ of the vertex-disjoint copies of~$H_1$ in the base
  graph). Thus,~$|S_A|\le k-q$. Finally, observe that since~$G-S_A$
  has a maximum matching of size~$q\le d-1$ it cannot contain a cycle
  as all cycles in~$G$ have length at least~$3d$. Thus,~$G-S_A$ has a
  vertex cover~$S_B$ of size at most~$q$ which together with~$S_A$ is
  a size-$k$ vertex cover of~$G$.

  Since the two instances are equivalent, any algorithm
  deciding~$(G',k')$ decides~$(G,k)$. Since~$H_\Pi$ has constant size for
  each fixed graph property~$\Pi$, the base graph contains~$O(|V(G)|)$
  vertices and edges. Similarly, the vc-extension of~$G$ has~$O(|V(G)|)$
  copies of~$J$ (since~$G$ has~$O(|V(G)|)$ edges) and~$O(|V(G)|)$ copies of~$D$
  and no further vertices. Consequently,~$G'$ has~$O(|V(G)|)$ edges and~$O(|V(G)|)$
  edges. Any algorithm deciding~$(G',k)$ in~$2^{o(|V(G')|+|E(G')|)}$ time thus
  decides~\sctvc{} in~$2^{o(|V(G)|)}$~time, contradicting the ETH by Theorem~\ref{thm:scvc}.
\end{proof}
Let us briefly discuss why we need to assume large girth for the \vc{}
instance. If~$\Pi$ is for example the property of being~$2K_2$-free,
that is, not containing an induced matching of size two, then given a
\vc{} instance~$(G,k)$, the reduction would simply add an edge
to~$G$. Now, if~$G$ is a triangle and~$k=1$, then we have a
no-instance but the resulting~\pivd{} instance is a yes-instance:
deleting the isolated edge is sufficient to destroy all~$2K_2$s.

Since the maximum degree is three in~\scvc{}, we obtain the
following corollary on the hardness of~\pivd{} in bounded-degree
graphs.
\begin{corollary}
  \label{cor:degree}
  Let~$\Pi$ be a hereditary nontrivial graph property such that all independent
  sets are in~$\Pi$. Let~$\Delta$ be the smallest number such that
  there exists a~$\gamma$-ordering of the forbidden subgraphs of~$\Pi$
  such that a forbidden subgraph with lexicographically
  smallest~$\Gamma$-sequence has maximum
  degree~$\Delta$. Then,~\pivd{} cannot be solved in~$2^{o(n+m)}$ time
  even if~$G$ has maximum degree~$3\Delta$.
\end{corollary}
So, for example if~$\Pi$ is the property of being acyclic,
then~\pivd{} (this special case is known as \textsc{Feedback Vertex
  Set}) does not admit a~$2^{o(n+m)}$-time algorithm even if~$G$ has
maximum degree six.

We can directly use Lemma~\ref{lem:no-ind-set} to obtain a lower bound
for all other nontrivial graph properties. First observe that, by
Ramsey's theorem, every nontrivial graph property~$\Pi$ contains
either all independent sets or all cliques: Since~$\Pi$ is nontrivial,
the order of the graphs in~$\Pi$ is unbounded. Thus, for every
number~$n$,~$\Pi$ contains a graph which has~$R_{n,n}$ many nodes and
therefore contains either an $n$-vertex clique or an~$n$-vertex
independent set. Thus, if~$\Pi$ does not contain all independent sets,
then it contains all cliques. In this case, however,
Lemma~\ref{lem:no-ind-set} implies a running time lower bound
for~$\bar{\Pi}$-\textsc{Vertex Deletion}, where~$\bar{\Pi}$ is the
graph property containing exactly the graphs that are complement
graphs of a graph in~$\Pi$. Now~$\bar{\Pi}$-\textsc{Vertex Deletion}
problem can be easily reduced to~\pivd{} by complementing the input
graph. This reduction does not change the number of vertices~$n$ in the
graph and the number of edges~$m$ is~$O(n^2)$,
thus leading to the following lower bound for all hereditary graph
properties.
\begin{corollary}
  Let~$\Pi$ be a hereditary graph property. Then, if the ETH is
  true,~\pivd{} cannot be solved in~$2^{o(n+\sqrt{m})}$ time.
\end{corollary}

\section{Subexponential-Time Algorithms for Graph Properties Excluding
  Some Independent Set}
\label{sec:subexp}
The hardness results from the previous section leave open the
possibility of algorithms that have subexponential running time with
respect to the number of edges~$m$ of~$G$ for graph properties~$\Pi$
that do not contain all independent sets.  We now show that~\pivd{}
can be solved in~$2^{o(m)}+O(n)$ time if~$\Pi$ can be recognized
sufficiently fast and does not contain all independent
sets. Intuitively, the algorithm exploits the following observations:
First, the number of vertices that have a high degree in~$G$ is
sublinear in~$m$.  Second, the subgraph induced by the low-degree
vertices in the solution has a small dominating set.
\begin{lemma}\label{lem:subexp}
  Let~$\Pi$ be a hereditary graph property such that for some fixed~$d$, $\Pi$
  does not contain the~$d$-vertex independent set, then
    \begin{itemize}
    \item \pivd{} can be solved in~$2^{O(\sqrt{m})} + O(n)$ time if
      and only if membership in~$\Pi$ can be verified in~$2^{O(n)}$
      time, and
    \item \pivd{} can be solved in~$2^{o(m)} + O(n)$ time if and only
      if membership in~$\Pi$ can be verified in~$2^{o(m)}$ time.
    \end{itemize}
\end{lemma}
\begin{proof}
  We first show an algorithm with the claimed running time for the
  case that membership in~$\Pi$ can be verified in~$2^{O(n)}$
  time. Given an input~$(G,k)$, the first step of the algorithm is to
  delete all except at most~$d-1$ singletons of the input graph and to
  decrease the size bound accordingly. These deletions are necessary
  by the assumption that~$\Pi$ does not contain the~$d$-vertex
  independent set. Afterwards, the graph has $O(m)$~vertices.
  
  Now, let~$A$ denote an arbitrary but fixed maximum-cardinality set
  such that~$G[A]\in \Pi$. Let~$V_h$ denote the vertices in~$G$ that
  have degree at least~$2\sqrt{m}$ and observe that~$|V_h|\le
  \sqrt{m}$. For each~$A_h\subseteq V_h$, branch
  into a case that assumes~$A_h=A\cap V_h$. In one of these branches,
  the assumption is correct. 

  Now consider the vertices in~$V_\ell:=V\setminus V_h$ that have
  degree at most~$2 \sqrt{m}$ in~$G$. Since~$G[A\cap V_\ell]\in \Pi$
  it has no independent set on~$d$ vertices. Thus, for each
  independent set~$I_\ell\subseteq V_\ell$ of~$G[V_\ell]$ such
  that~$|I_\ell|< d$ we branch into a case that assumes
  that~$I_\ell\subseteq A\cap V_\ell$ is a maximal independent set
  of~$G[A\cap V_\ell]$. In one of these branches, the assumption is
  correct. Now observe that since~$I_\ell$ is a \emph{maximal}
  independent set in~$G[A\cap V_\ell]$ it is also a dominating set
  in~$G[A\cap V_\ell]$. That is,~$A\cap V_\ell\subseteq
  N[I_\ell]$. Thus, by branching for each~$A_\ell\subseteq N(I_\ell)$
  into the case that~$A_\ell=(A\cap V_\ell)\setminus I_\ell$ we obtain
  one case where~$A=A_h\cup I_\ell\cup A_\ell$. In each branch, we check
  whether~$G[A_h\cup I_\ell\cup A_\ell]\in \Pi$ and whether~$|A_h\cup
  I_\ell\cup A_\ell|\ge n-k$.
  
  The running time bound can be obtained as follows: The number of
  subsets of~$V_h$ is~$2^{O(\sqrt{m})}$. The number of
  subsets~$I_\ell$ of size at most~$d$ of~$V_\ell$ is
  $O(n^{d-1})=n^{O(1)}$, and for each~$I_\ell$ the number of subsets
  of~$N(I_\ell)$ is~$O(2^{(d-1)2\sqrt{m}})=2^{O(\sqrt{m})}$. Thus, the
  recognition algorithm for~$\Pi$ has to be invoked~$2^{O(\sqrt{m})}$
  times. Each time, it is invoked on a graph with~$O(\sqrt{m})$
  vertices since~$|A_h|\le \sqrt{m}$,~$|I_\ell|< d=O(1)$, and~$|A_\ell|
  < d\sqrt{m} =O(\sqrt{m})$. By the premise, this algorithm
  takes~$2^{O(\sqrt{m})}$ time, resulting in the claimed overall
  running time.

  For the only if part of the first claim, we observe the following.
  Since~$m=O(n^2)$, any~$2^{O(\sqrt{m})}$-time algorithm for~\pivd{}
  directly gives a~$2^{O(n)}$-time algorithm for the recognition
  problem, since the special case~$k=0$ is the recognition problem
  for~$\Pi$.

  If membership in~$\Pi$ can be verified in~$2^{o(m)}$ time, then the
  algorithm only differs in the invoked recognition algorithm. This
  algorithm for~$\Pi$ has to be invoked~$2^{O(\sqrt{m})}=2^{o(m)}$ times, each time on a graph with at most~$m$
  edges. By the premise, this algorithm takes~$2^{o(m)}$ time,
  resulting in the claimed overall running time. The only if part of
  the statement follows directly from the fact that solving the
  special case~$k=0$ gives an algorithm for the recognition problem.
\end{proof}
Since the time complexity of~\pivd{} for the nontrivial graph
properties excluding a fixed independent is settled to
be~$2^{O(\sqrt{m})}$ or more, it is now motivated to decrease the
constants in the exponents of the running time bound. In the
following, we consider the important special case of graph properties
where membership can be decided in polynomial time. The generic
algorithm described above has running time~$2^{\sqrt{m}}\cdot
n^{d-1}\cdot 2^{(d-1)2\sqrt{m}}\cdot n^{O(1)}$ for these problems. By
setting the degree threshold for including a vertex in the first or in
the second branching to~$\sqrt{2m/(d-1)}$, this can be improved to a
running time of~$2^{2\sqrt{2(d-1)m}}n^{d + O(1)}$. It is more
interesting to determine whether the factor~$d$ before~$m$ can be
removed. For a subclass of the polynomial-time decidable graph
properties excluding an independent set of size~$d$, we obtain such a
running time bound. The algorithm achieving the running time follows
the same idea as a known~$2^{O(\sqrt{m})}$-time algorithm for
\textsc{Clique}~\cite{FK10}.
\begin{theorem}\label{thm:degree-subexp}
  Let~$\Pi$ be a graph property such that
  \begin{itemize}
  \item all~$n$-vertex graphs~$G$ with property~$\Pi$ have minimum degree~$n-d$ for some constant~$d$, and 
  \item for an~$n$-vertex graph, membership in~$\Pi$ can be verified
    in~$n^{O(1)}$ time.
  \end{itemize}
  Then,~\textsc{$\Pi$-Vertex Deletion} can be solved
  in~$2^{\sqrt{2m}}\cdot n^{d+O(1)}$ time.
\end{theorem} 
\appendixproof{Theorem~\ref{thm:degree-subexp}}{
\begin{proof}
  Let~$A$ denote an arbitrary but fixed maximum-cardinality set
  such that~$G[A]\in \Pi$. First, consider the case that~$A$ contains
  a vertex~$v$ that has degree at most~$\sqrt{2m}$ in~$G$. Then,~for
  each~$A'\subseteq N(v)$ branch into a case assuming that~$A'=A\cap
  N(v)$. In each of these cases, branch for each subset~$A_d$ of size
  at most~$d$ of~$V\setminus N[v]$ into a case
  assuming~$A_d=A\setminus N[v]$. Now if~$|A'\cup A_d\cup \{v\}|\ge
  n-k$, then check whether~$G[A'\cup A_d\cup \{v\}]$ fulfills~$\Pi$
  using the algorithm promised by the premise of the lemma. If this is the case,  then return ``yes''. The overall number of branches is~$2^{\sqrt{2m}}\cdot n^d$, resulting in a running time of~$2^{\sqrt{2m}}\cdot n^{d+O(1)}$ for this part of the algorithm. If~$A$ contains~$v$, then in one of the created
  branches both assumptions are correct. Hence, if all of the branches
  for all~$v$ with degree at most~$\sqrt{2m}$ return ``no'', then
  $A$~contains only vertices of degree at least~$\sqrt{2m}+1$. Thus, all
  other vertices may be removed from the graph while decreasing the size
  bound~$k$ accordingly. 

  In the remaining graph,~$n < \sqrt{2m}$, thus a brute-force
  algorithm testing for each vertex subset~$A$ of size at least~$n-k$
  whether~$G[A]$ fulfills~$\Pi$ has running
  time~$2^{\sqrt{2m}}\cdot n^{O(1)}$.
\end{proof}
}

\section{Consequences for Some Restricted Domains of Input Graphs}
\label{sec:restricted}
We now present running time lower bounds for \pivd{} when the input
graph is restricted to belong to a certain graph class. Such
restrictions are motivated by at least two aspects: From an application
viewpoint, one should provide hardness results for realistic types
of input instances, for example for sparse graphs. From a
complexity-theoretic viewpoint, hardness results for restricted inputs
may facilitate further reductions.

\subsection{Planar Graphs}
\label{sec:planar}
\decprob{\ppivd}{An undirected planar graph~$G=(V,E)$ and an
  integer~$k$.}{Is there a set~$S\subseteq V$ such that~$|S|\le k$
  and~$G[V\setminus S]$ is contained in~$\Pi$?}  In \ppivd, we are
only interested in properties~$\Pi$ that contain all independent sets:
Otherwise, for some fixed~$d\ge 5$ (depending on~$\Pi$), the
property~$\Pi$ cannot contain planar graphs of order at
least~$R_{d,d}$ since these graphs cannot contain an independent set
of size~$d$ and thus contain a~$K_5$.

For graph properties~$\Pi$ containing all independent sets and
excluding at least one planar graph, we can apply our modification of
Yannakakis' reduction when we reduce from \pvc{} instead: First, as
noted for example in~\cite[Theorem 14.9]{CFK+15} it is known that
\textsc{Planar Vertex Cover} cannot be solved in~$2^{o(\sqrt{n})}$
time (assuming the ETH).\footnote{The result follows essentially from
  a classic reduction~\cite{GJS76} of \vc{} to \pvc{} that uses
  constant-size uncrossing gadgets.} Second, the reduction behind
Theorem~\ref{thm:scvc} produces a planar graph if it has a planar
graph as input. Third, if~$\Pi$ is hereditary and excludes some planar
graph, then it has a nonempty family of planar forbidden induced
subgraphs~${\cal F}$. Finally, in a reduction producing a planar
graph, we may ignore all nonplanar forbidden induced subgraphs. Thus,
we can simply carry out the reduction behind
Lemma~\ref{lem:no-ind-set} starting from \pvc{} and picking the gadget
graph~$H_\Pi$ only among the planar forbidden induced subgraphs
of~$\Pi$. The replacement of edges of the \pvc{} instance by planar
graphs and the attachment of planar graphs to single vertices clearly
yields a planar graph, as noted by Yannakakis~\cite{LY80}. Moreover,
if the \pvc{} instance has~$O(n)$ vertices, then so has the~\ppivd{}
instance.  Altogether, we arrive at the following.
\begin{theorem}\label{thm:planar}
  Let~$\Pi$ be a hereditary graph property containing all independent
  sets and excluding at least one planar graph. Then, if the ETH is
  true,~\ppivd{} cannot be solved in~$2^{o(\sqrt{n})}$ time.
\end{theorem}

\subsection{Bounded-Degeneracy Graphs}
\label{sec:degeneracy}
For the bounded-degeneracy case, we first observe
that~Corollary~\ref{cor:degree} already implies that for~$\Pi$
containing all independent sets,~\pivd{} deletion cannot be solved in
subexponential time even on graphs with bounded degree and thus not on
graphs with bounded degeneracy. Here, we improve the bound
on the degeneracy. 

\begin{theorem}\label{thm:degen}
  Let~$\Pi$ be a hereditary graph property such that all independent
  sets are in~$\Pi$. Let~$\delta$ be the smallest number such that there
  exists a~$\gamma$-ordering of the forbidden subgraphs of~$\Pi$ such
  that a forbidden subgraph with lexicographically
  smallest~$\Gamma$-sequence is~$\delta$-degenerate. Then,~\pivd{} cannot
  be solved in~$2^{o(n+m)}$ time even if~$G$ is~$(\delta+1)$-degenerate.
\end{theorem}
\appendixproof{Theorem~\ref{thm:degen}}{
\begin{proof}
  We modify Construction~\ref{cons:main} as follows. Let~$H_\Pi$
  denote the~$\delta$-degenerate forbidden subgraph of~$\Pi$ that has
  the lexicographically smallest~$\Gamma$-sequence among all forbidden
  induced subgraphs of~$\Pi$. As before, let~$H_1$ denote the
  connected component of~$H_\Pi$ that has the highest~$\gamma$-value
  and fix~$c$ again to be a vertex such that~$\alpha(H,c)=\gamma(H)$
  and let~$J$ denote the induced subgraph containing~$c$ and an
  arbitrary but fixed largest connected component
  of~$H_1-c$. Let~$n_J$ denote the order of~$J$ and fix a
  degeneracy-ordering of~$J$, that is, a sequence~$(j_1,\ldots
  ,j_{n_J})$ such that~$j_i$ has degree at most~$\delta$ in~$J[\{j_i,
  j_{i+1},\ldots , j_{n_J}\}]$. Now let~$c'$ denote the vertex with the
  highest index that is different from~$c$. More precisely, if~$c\neq
  j_{n_J}$, then $c'=j_{n_J}$; otherwise $c'=j_{n_J-1}$.

  Now perform the construction as before using the new choice of~$c$
  and~$c'$ when replacing an edge of the vertex cover instance by a
  copy of~$J$. All other parts of the construction remain the
  same. Observe that in this modified construction we only make the
  choice of~$c'$ specific where it was arbitrary before. This implies
  that the modified construction is correct and that any algorithm
  solving~\pivd{} on the constructed instances~$(G',k)$ implies
  a~$2^{o(n+m)}$-time algorithm for~\scvc{}, violating ETH.

  It remains to show that the constructed instance
  is~$(\delta+1)$-degenerate. The base graph of the construction is
  clearly~$\delta$-degenerate as each connected component is a subgraph
  of~$H_\Pi$.

  We now describe a sequence of vertex deletions for the vc-extension
  such that each deleted vertex has degree at most~$\delta+1$ when it
  is deleted. First, observe that the graph~$D:=H_1 - (V(J)\setminus
  \{c\})$ which is added for each vertex~$v$ of the original instance
  is~$\delta$-degenerate. In every copy~$D^*$ of~$D$, all vertices
  except~$v$ (which is the vertex identified with the cut-vertex~$c$)
  have no neighbors outside of this copy. Since~$D-c$
  is~$\delta$-degenerate, every induced subgraph of~$D^*$ that
  contains~$v$ plus some other vertex~$u\neq v$, thus contains a
  vertex different from~$v$ that has degree at
  most~$\delta+1$. Consequently, this vertex can be deleted
  first. This vertex deletion can be performed as long as any
  copy~$D^*$ of~$D$ still contains a vertex different from the vertex
  that was identified with~$c$. 

  The remaining graph consists of copies~$J^*$ of~$J$ that have two
  vertices which are identified with vertices~$V$, the vertex set of
  the \sctvc{}~instance. Consider a such a copy~$J^*$ replacing an
  edge~$\{u,v\}$, that is, the vertices~$c$ and~$c'$ are identified
  with~$u$ and~$v$. By the modification of the construction,
  either~$c$ or~$c'$ equals~$j_{n_J}$. As a consequence, deleting all
  vertices of~$J^*$ that are different from~$u$ and~$v$ in the same
  order as before, gives a deletion sequence in which each deleted
  vertex has degree at most~$\delta+1$. After deleting all these
  vertices in each copy of~$J$, what remains is either an independent
  set (if~$c$ and~$c'$ are not adjacent in~$J$) or exactly the graph~$G$
  of the original \sctvc{} instance. This graph
  is~$2$-degenerate. Since~$\Pi$ contains all independent
  sets,~$H_\Pi$ is not $0$-degenerate and thus~$2\le (\delta+1)$.
\end{proof}
} An example application of Theorem~\ref{thm:degen} is the
following. If all forbidden subgraphs of~$\Pi$ are
acyclic, then~\pivd{} cannot be solved in~$2^{o(n+m)}$ time even
if~$G$ is 2-degenerate.
\subsection{Graphs with a Dominating Vertex}
\label{sec:dominating}

Next, we consider instances which have a dominating vertex and thus
diameter two. 
\begin{proposition}\label{prop:dominating}
  Let~$\Pi$ be a hereditary graph property.
    If~$\Pi$ contains all independent sets, then $\pivd{}$ cannot
    be solved in~$2^{o(n+m)}$ time even if~$G$ has a dominating
    vertex.
\end{proposition}
\appendixproof{Proposition~\ref{prop:dominating}}{
\begin{proof}
  We extend Construction~\ref{cons:main}. Recall that~$(G,k)$ is the
  \sctvc{} instance and that~$(G',k)$ is the instance of~\pivd{} obtained
  by Construction~\ref{cons:main}. Let~$H_\Pi$ denote again the
  forbidden induced subgraph of~$\Pi$ that has the lexicographically
  smallest~$\Gamma$-sequence, let~$H_1$ denote the connected component
  of~$H_\Pi$ that has maximum~$\gamma$-value and assume that~$d$
  connected components of~$H_\Pi$ are isomorphic to~$H_1$. Finally,
  let~${\cal H}^{<}$ denote the set of graphs whose~$\gamma$-value is
  smaller than~$H_1$. Now distinguish two cases.

  \emph{Case~1: There is a graph~${\cal I}$ whose connected
    components~$I_1, \ldots ,I_q$ are all from ${\cal H}^{<}$ such
    that the graph obtained by
    \begin{itemize}
    \item taking the disjoint union of~${\cal I}$
      and~$d-1$ copies of~$H_1$, and then
    \item adding a further vertex~$v$ and making~$v$ adjacent to all
      vertices of the graph
    \end{itemize}
    is not contained in~$\Pi$.} Observe that by the choice of~$H_\Pi$
  and~${\cal H}^{<}$,~${\cal I}$ is contained in~$\Pi$.  Perform the
  construction as previously, let~$(G',k)$ denote the graph as
  previously constructed. Now take the disjoint union of~${\cal I}$
  and~$G'$. Then, add a vertex~$v^*$ and make it adjacent to all
  vertices of this graph and call the resulting graph~$G^*$. We show
  that~$(G^*,k+1)$ and~$(G',k)$ are equivalent instances.

  First, assume that~$G'$ has a~$\Pi$-vertex deletion set~$S'$ of size
  at most~$k$. The proof of Lemma~\ref{lem:no-ind-set} shows that~$G$
  has a vertex cover of size~$k$ which then implies that we can assume
  without loss of generality that~$S'$ does not contain vertices of the
  base graph. Thus,~$S'$ contains~$d-1$ connected components
  that are isomorphic to~$H_1$ and all other connected components
  of~$G'-S'$ have smaller~$\gamma$-value than~$H_1$. By the choice
  of~$H_\Pi$ and by the fact that all connected components of~${\cal
    I}$ are from~${\cal H}^<$, we have that the disjoint union
  of~$G-S$ and~${\cal I}$ is contained
  in~$\Pi$. Thus,~$G^*-(S'\cup\{v^*\})$ is contained in~$\Pi$ and~$G^*$
  has a~$\Pi$-vertex deletion set of size at most~$k+1$.

  Now assume that~$(G',k)$ is a no-instance, that is,~$G'$ has
  no~$\Pi$-vertex deletion set of size at most~$k$. Let~$S'$ be a
  minimum-cardinality~$\Pi$-vertex deletion set of~$G'$. If~$|S'|\ge
  k+2$, then any~$\Pi$-vertex deletion set of~$G^*$ has size at
  least~$k+2$ and the instances are equivalent. Hence,
  assume~$|S|=k+1$. We first show that~$S$ of~$G$ leaves exactly~$d-1$
  disjoint copies of~$H_1$ in~$G'-S$. Assume towards a contradiction
  that there is a minimum-cardinality~$\Pi$-vertex deletion set~$S$ such
  that~$d'<d-1$ disjoint copies of~$H_1$ remain in~$G-S$. Thus, at
  least~$d^*=(d-1)-d'$ copies of~$H_1$ in the base graph are destroyed
  which means that~$S'$ contains at most~$k+1-d^*$ vertices of the
  vc-extension. Furthermore, observe that the vc-extension of~$G'-S'$
  contains at most~$d^*-1$ disjoint copies of~$H_1$ by the assumption
  on~$S'$. We show that~$G$ has a vertex cover of size at most~$k-1$.

  To this end, create a set~$S_A\subseteq V$ as follows:
  \begin{itemize}
  \item For each vertex~$v\in V\cap S'$, add~$v$ to~$S_A$,
  \item for each copy~$D^*$ of~$D$ in~$G'$ such that~$S'$
    contains a vertex from~$V(D^*)\setminus
    V$, add the vertex~$v\in V(D^*)\cap
    V$ to~$S_A$,
  \item for each copy~$J^*$ of~$J$ such that~$S'\setminus V$ contains
    a vertex from~$V(J^*)\setminus V$ add an arbitrary vertex
    of~$V(J^*)\cap V$ to~$S_A$.
  \end{itemize}
  The set~$S_A$ has size at most~$k+1-d^*$. Now, consider the
  graph~$G-S_A$ and observe that every edge~$\{u,v\}$ in~$G-S_A$
  directly corresponds to an induced copy of~$H_1$ in~$G'-S'$:
  Since~$\{u,v\}$ is present in~$G-S_A$, we have that~$S'$ does not
  contain~$u$ or~$v$, does not contain vertices from the copies of~$D$
  that are attached to~$u$ or~$v$, and does not contain a vertex from
  the copy of~$J$ attached to~$u$ and~$v$. Now let~$q$ denote the size
  of a maximum matching in~$G-S_A$ and observe that this implies that
  the vc-extension of~$G'-S'$ has~$q$ vertex-disjoint copies
  of~$H_1$. Hence,~$q<d^*<d$. Thus,~$G-S_A$ has a maximum matching of
  size less than~$d$ and thus (since~$G$ has girth~$3d$) a vertex
  cover of size~$q<d^*$. Thus,~$G$ has a vertex cover of
  size~$k+1-d^*+q\le k$. This implies that~$G'$ has a~$\Pi$-vertex
  deletion set of size at most~$k$, contradicting our assumption
  on~$S$.

  Thus, if~$S$ has minimum-cardinality, then~$G^*-S$ contains
  exactly~$d-1$ disjoint copies of~$H_1$. Now~$G^*-S$ contains a
  disjoint union of~${\cal I}$ and~$d-1$ connected components
  isomorphic to~$H_1$ and a vertex~$v^*$ that is adjacent to all
  vertices in~$G^*-S$. Hence,~any~$\Pi$-vertex deletion set of~$G^*$
  has size at least~$k+2$.

  \emph{Case~2: otherwise.} Add a vertex~$v$ to the graph~$G'$
  obtained by the original construction and make~$v^*$ adjacent to all
  vertices in~$G$, call the resulting graph~$G^*$. We show
  that~$(G^*,k)$ and~$(G',k)$ are equivalent. First, if~$G^*$ has
  a~$\Pi$-vertex deletion set of size at most~$k$, then so does~$G'$
  since~$G'$ is an induced subgraph of~$G^*$. Second, let~$S'$ denote
  a~$\Pi$-vertex deletion set of size at most~$k$. As discussed in the
  proof of Case~1, we can assume without loss of
  generality that~$S'$ does not contain vertices of the base graph.

  Consequently,~$G'-S'$ contains a disjoint union of~$d-1$ copies
  of~$H_1$ and further graphs~$H'$
  with~$\gamma(H')<\gamma(H_1)$. Thus, with the exception of the
  copies of~$H_1$, every other connected component of~$G'-S'$ is a graph
  from~${\cal H}^-$. By the case assumption, adding a vertex to~$G'-S'$
  and making it adjacent to all vertices of~$G'-S'$ gives a graph
  in~$\Pi$. Hence, the isomorphic graph~$G^*-S'$ is also contained
  in~$\Pi$.

  Summarizing, in both cases we can modify the construction such that the
  resulting graph has one dominating vertex. Observe that this
  increases the number of edges by at most~$|V(G')+O(1)|$. Thus, the
  resulting graph has~$O(|V(G)|)$ edges and vertices, the running time
  bound follows.
\end{proof}}
Now for graph properties excluding some fixed independent set, all we
can hope for is excluding~$2^{o(n+\sqrt{m})}$-time algorithms (since the
general case can be solved in~$2^{o(m)}$ time).
\begin{proposition}\label{prop:dominating2}
  Let~$\Pi$ be a hereditary graph property such that~$\Pi$ does not
  contain all independent sets, then~$\pivd{}$ cannot be solved
  in~$2^{o(n+\sqrt{m})}$ time even if~$G$ has a dominating vertex.
\end{proposition}
\appendixproof{Proposition~\ref{prop:dominating2}}
{
\begin{proof}
  Let~$\bar{\Pi}$ denote the graph property that contains all graphs
  that are complement graphs of graphs in~$\Pi$.  Consider
  Construction~\ref{cons:main} when applied to show that~$\bar{\Pi}$
  cannot be solved in~$2^{o(n)}$ time. Let~$G'$ denote the graph
  obtained by this construction.  Since~$H_1$ contains at least one
  edge, one may safely add an isolated vertex at the end of the
  construction and obtain a graph~$G^*$ that has a $\bar{\Pi}$-vertex
  deletion set of size at least~$k$ if and only if~$G'$
  does. Hence,~$\bar{\Pi}$-vertex cannot be solved in~$2^{o(n+m)}$ time
  assuming the ETH, also if the input graph~$G^*$ has an isolated
  vertex. Now the claimed hardness result follows from the fact that
  the complement graph~$\bar{G^*}$ has a dominating vertex and the
  same number of overall vertices and every induced
  subgraph~$\bar{G^*}[X]$ of~$\bar{G^*}$ is in~$\Pi$ if and only
  if~$G^*[X]$ is in~$\bar{\Pi}$.
\end{proof}
}

\section{Connected $\mathbf{\Pi}$-Vertex Deletion}
Finally, we consider a popular variant of~\pivd{} where the set of
deleted vertices has to be connected. Special cases of this problem
that have been studied include \textsc{Connected Vertex
  Cover}~\cite{GNW07}, \textsc{Connected Feedback Vertex
  Set}~\cite{GS09,MPRSS12}, and \textsc{Minimum $k$-Path Connected
  Vertex Cover}~\cite{LLWW13}.  \decprob{\cpivd}{An undirected
  graph~$G=(V,E)$ and an integer~$k$.}{Is there a set~$S\subseteq V$
  such that~$G[S]$ is connected,~$|S|\le k$, and~$G[V\setminus S]$ is
  contained in~$\Pi$?}  For this problem, we obtain the same running
time bounds as for \pivd{}.
\begin{theorem}\label{thm:main-con}
  Let~$\Pi$ be a hereditary nontrivial graph property, then:
  \begin{enumerate}
  \item If the ETH is true, then \cpivd{} cannot be solved in~$2^{o(n+\sqrt{m})}$ time,
  \item If the ETH is true, then \cpivd{} cannot be solved in~$2^{o(n+m)}$ time if~$\Pi$ contains
    all independent sets.
  \item If~$\Pi$ excludes some independent set, then \cpivd{} can be
    solved in~$2^{O(\sqrt{m})}+O(n)$ time if and only if membership
    in~$\Pi$ can be recognized in~$2^{O(n)}$ time.
  \item If~$\Pi$ excludes some independent set, then \cpivd{} can be
    solved in~$2^{o(m)}+O(n)$ time if and only if membership in~$\Pi$
    can be recognized in~$2^{o(m)}$ time.
  \end{enumerate}
\end{theorem}
The two positive results for~$\Pi$ excluding some independent set can
be easily obtained by adapting the algorithm described in
Lemma~\ref{lem:subexp} in such a way that solutions (vertex sets to
delete) are only accepted if they induce a connected subgraph. Also,
the only if part of the last two statements follows directly since it is
obtained by considering the special case~$k=0$, where the solution is
empty and thus trivially connected.  To prove the running time lower
bounds given in Theorem~\ref{thm:main-con}, we extend
Construction~\ref{cons:main}. One approach could be to add a universal
vertex that needs to be deleted. To obtain hardness also in the case
of bounded-degree graphs (if~$\Pi$ contains all independent sets), we
instead modify the construction by adding a set of vertices that is
part of every connected minimum-cardinality $\Pi$-vertex deletion set
and connects the solution vertices of the original instance by forming
a binary tree.
\begin{lemma}\label{lem:con-hard}
  Let~$\Pi$ be a hereditary nontrivial graph property containing all
  independent sets. If the ETH is true, then \cpivd{} cannot be solved
  in~$2^{o(n+m)}$ time.
\end{lemma}
\appendixproof{Lemma~\ref{lem:con-hard}}
{
\begin{proof}
  We reduce from \sctvc{} by adapting
  Construction~\ref{cons:main}. Let~$(G=(V,E),k)$ denote the instance
  of~\scvc{}. To simplify the construction somewhat, assume without
  loss of generality that~$V=\{1,\ldots, n\}$ and that the number~$n$
  of vertices in~$G$ is a power of~$2$, that is,~$n=2^\mu$ for some~$
  \mu \in \mathds{N}$. First, perform Construction~\ref{cons:main}.
  Then, set~$i:= \mu-1$ and add a set~$V_i:=\{v_i^1,\ldots,v_{2^i}\}$
  of~$2^i$ vertices, and make each~$v_i^j$ adjacent to~$j$ and~$j+2^i$
  (recall that the output graph of Construction~\ref{cons:main}
  has~$V$ as a vertex subset). Now, for each vertex~$v_i^j$ of~$V_i$,
  add a copy of~$H_1$ and identify~$c$ (a fixed vertex of~$H_1$ such
  that~$\alpha(H_1,c)=\alpha(H_1)$) with~$v_i^j$.

  Continue the construction for decreasing~$i$, that is, until~$i=0$,
  do the following. Set~$i\leftarrow i-1$, and add a vertex
  set~$V_i:=\{v_i^1,\ldots,v_i^{2^i}\}$. Then, for each~$v_i^j$,
  make~$v_i^j$ adjacent to~$v_{i+1}^j$ and~$v_{i+1}^{j+2^i}$, and add
  a copy of~$H_1$ and identify~$c$ with~$v_i^j$. Let~$V_T :=
  \bigcup_{i\in [\mu]} V_i$ and call the copies of~$H_1$ that are
  added in this step the~$V_T$-attachments. Let~$G':=(V',E')$ denote
  the graph obtained this way. Conclude the reduction by
  setting~$k':=k+|V_T|$. Observe that~$|V'|=O(n)$
  and~$|E'|=O(n)$. Thus, to show the lemma, it remains to show the
  equivalence of the \sctvc{} instance~$(G,k)$ and of the \cpivd{}
  instance~$(G',k')$.

  If~$(G,k)$ is a yes-instance with a size-$k$ vertex cover~$S$,
  then~$S':= V_T\cup S$ is a size-$k$ $\Pi$-vertex
  deletion set: the graph~$G'-S'$ has at most~$d-1$ connected
  components with~$\gamma$-value~$\gamma(H_1)$ and all other connected
  components have~$\gamma$-value less than~$\gamma(H_1)$ (recall
  that~$d$ is the number of connected components of~$H_\Pi$ that are
  isomorphic to~$H_1$). Then, by the definition of~$H_\Pi$,~$G'-S'$
  fulfills~$\Pi$. Moreover,~$G'[S']$ is connected: $G[S'\setminus
  S]$ is a binary tree and every vertex of~$V$ has one neighbor
  in~$S'\setminus S$. 

  For the converse, observe that every connected $\Pi$-vertex deletion
  set~$S'$ deletes at least one vertex in the vc-extension. Thus, it may
  only delete vertices of the vc-extension and
  of~$V_T$-attachments. Consequently, only the base graph may contain
  induced subgraphs that are isomorphic to~$H_1$. Note that
  each~$V_T$-attachment is isomorphic to~$H_1$. Since, there are~$|V_T|$ vertex-disjoint $V_T$-attachments, at least~$|V_T|$ vertices of~$S'$ are not from the vc-extension. This implies that there is
  a set~$S^*$ of at most~$k'-|V_T|=k$ vertices in the vc-extension whose
  deletion destroys all induced subgraphs of the vc-extension that are
  isomorphic to~$H_1$. As in the the proof of
  Lemma~\ref{lem:no-ind-set}, this implies that~$G$ has a vertex cover
  of size at most~$k$.
\end{proof}}
The reduction behind Lemma~\ref{lem:con-hard} is an extension of
Construction~\ref{cons:main} which increases the vertex degree of
every vertex by at most one and additionally adds vertices whose
degree is at most~$\Delta(H_\Pi)+3$, where~$\Delta(H_\Pi)$ is the
maximum degree in~$H_\Pi$.
\begin{corollary}
  \label{cor:conn-degree}
  Let~$\Pi$ be a hereditary nontrivial graph property such that all independent
  sets are in~$\Pi$. Let~$\Delta$ be the smallest number such that
  there exists a~$\gamma$-ordering of the forbidden subgraphs of~$\Pi$
  such that a forbidden subgraph with lexicographically
  smallest~$\Gamma$-sequence has maximum
  degree~$\Delta$. Then,~\cpivd{} cannot be solved in~$2^{o(n+m)}$ time
  even if~$G$ has maximum degree~$3\Delta+1$.
\end{corollary}
For properties~$\Pi$ that do not include all independent sets, we also obtain a tight result.
\begin{lemma}\label{lem:con-hard-2}
  Let~$\Pi$ be a hereditary nontrivial graph property excluding some
  independent set. If the ETH is true, then \cpivd{} cannot be solved
  in~$2^{o(n+\sqrt{m})}$ time.
\end{lemma}
\appendixproof{Lemma~\ref{lem:con-hard-2}}
{
\begin{proof}
  Let~$\bar{\Pi}$ denote the property that contains all graphs that
  are complement graphs of a graph in~$\Pi$. Reduce from \sctvc{} by
  first performing Construction~\ref{cons:main} with one difference: 
  add~$d$ (instead of~$d-1$) copies of~$H_1$ to the base
  graph. Let~$G'$ denote the graph obtained by this
  construction. Observe that at least one vertex of one of these~$d$
  copies of~$H_1$ is contained in any~$\bar{\Pi}$-vertex deletion set
  of the resulting graph. Moreover, it is sufficient to delete only
  one vertex~$v$ in this copy of~$H_1$. Thus,~$G'$ has
  a~$\bar{\Pi}$-vertex deletion set of size at most~$k+1$ if and only
  if~$G$ has a vertex cover of size at most~$k$. Now, let~$\bar{G'}$
  denote the complement graph of the graph obtained by the
  modified~Construction~\ref{cons:main}. There is a
  minimum-cardinality $\Pi$-vertex deletion set that contains~$v$ and
  no further vertex from the copy of~$H_1$ that
  contains~$v$. Hence,~$v$ is in~$\bar{G'}$ adjacent to all other
  vertices of this minimum-cardinality $\Pi$-vertex deletion set,
  which makes the subgraph induced by this vertex deletion set
  connected. Therefore,~$\bar{G'}$ has a connected $\Pi$-vertex
  deletion set of size at most~$k$ if and only if~$G$ has a vertex
  cover of size at most~$k$. The running time bound follows from
  observing that~$\bar{G'}$ has~$O(|V|^2)$ edges and~$O(|V|)$
  vertices where~$V$ is the vertex set of the \sctvc{} instance.
\end{proof}
} Finally, let~\cppivd{} denote the variant of~\cpivd{} where the
input graph is restricted to be planar. We can show the same running
time lower bound as for the unconstrained vertex deletion problem on
planar graphs.
\begin{theorem}
  \label{thm:con-hard-planar} Let~$\Pi$ be a hereditary graph property
  containing all independent sets and excluding at least one planar
  graph. Then, if the ETH is true,~\cppivd{} cannot be solved
  in~$2^{o(\sqrt{n})}$ time.

\end{theorem}
\appendixproof{Thm~\ref{thm:con-hard-planar}} { \begin{proof}
  Reduce from \textsc{Planar Vertex Cover}. Given an
    instance~$(G,k)$ of \textsc{Planar Vertex Cover}, fix an arbitrary
    embedding of~$G$. Now add one vertex into every face and make this
    vertex adjacent to all vertices on the boundary of the face. Call
    the additional vertex set~$V_F$ and let~$E_F$ denote set of edges
    incident with vertices of~$V_F$. Now, perform
    Construction~\ref{cons:main} on~$G$ as before, that is, add a base
    graph, replace each edge of~$G$ by a copy of~$J$ and add for each
    vertex~$v$ of~$G$ a copy of~$H_1-(V(J)\setminus \{c\})$ and
    identify~$c$ with~$v$. Now, for each vertex~$v\in V_F$ add a copy
    of~$H_1$ and identify~$c$ with~$v$. Call the resulting graph~$G'$
    and set~$k':=k+|V_F|$. Observe that~$|V_F|=O(|V|)$ and
    that~$|E_F|=|V_F|$ (we may assume without loss of generality
    that~$G$ is connected). Moreover,~$G'$ is planar, it is obtained
    from a planar graph by replacing edges and vertices by planar
    graphs. Thus, if the two instances are equivalent, then any~$2^{o(\sqrt{n})}$-time algorithm for \cppivd{} implies a~$2^{o(\sqrt{n})}$ for \textsc{Planar Vertex Cover}. To complete the proof, we thus need to show that~$G$ has a
    vertex cover of size~$k$ if and only if~$G'$ has a connected
    $\Pi$-vertex deletion set of size~$k'$.

  $\Rightarrow$: If~$S$ is a size-$k$ vertex cover of~$G$, then~$S\cup
  V_F$ is a connected size-$k$ $\Pi$-vertex deletion set of~$G'$:
  After the deletion of~$S\cup V_F$ all connected components of~$G'$
  have~$\gamma$-value lower than~$H_1$ and thus, by the definition
  of~$H_\Pi$ and the construction of the base graph, the remaining
  graph is contained in~$\Pi$. It remains to show that~$G'[S\cup V_F]$
  is connected. Every vertex in~$S$ has a neighbor in~$V_F$. Moreover,
  every pair of vertices~$u,v\in V_F$ that were placed into faces
  of~$G$ whose boundaries share an edge are connected in~$G'[S\cup
  V_F]$ because one of the two endpoints of this edge is contained
  in~$S$ and both endpoints are neighbors of~$u$ and~$v$. Thus, all
  vertices of~$V_F$ are connected in~$G'[S\cup V_F]$ 
  since the dual graph of the planar graph~$G$ is connected.

  $\Leftarrow$: Let~$S'$ denote a connected~$\Pi$-vertex deletion set
  of size at most~$k'$ in~$G'$. Since~$G'[S']$ is connected,~$S'$
  cannot contain any vertices of the base graph. Consequently, by the
  construction of the base graph and the choice of~$\Pi$, all induced
  subgraphs isomorphic to~$H_1$ must be destroyed in the connected
  component containing the vc-extension. Thus, for each copy of~$H_1$
  that was attached to a vertex of~$V_F$ at least one vertex must be
  deleted. This implies that~$S'$ contains at most~$k$ vertices of
  the vc-extension. Observe that for each subgraph isomorphic to~$H_1$
  in the vc-extension, there are two vertices~$\{u,v\}$ that are from~$V$ and
  deleting these two vertices cuts the other vertices from the rest of
  the graph. Thus, since~$G'[S']$ is connected at least one of these
  two vertices is contained in~$S'$. Let~$S$ denote the set of these
  at most~$k$ vertices. Since for each edge~$\{u,v\}$ of~$G$, the graph~$G'$ contains a such subgraph isomorphic to~$H_1$ having~$u$ and~$v$ as cut-vertices,~$S$ is a vertex cover in~$G$.
\end{proof}
}
\bibliographystyle{abbrv} 
\bibliography{eth-bound}



\end{document}